\documentclass[twocolumn]{IEEEtran}
\usepackage{graphicx}
\usepackage{amsmath,amssymb} 
\usepackage{ifthen}
\usepackage{setspace}
\usepackage{accents}
\usepackage[latin1]{inputenc}
\usepackage[pdftex]{hyperref}

\newtheorem{theorem}{\bf\em Theorem}
\newtheorem{lemma}{\bf\em Lemma}
\newtheorem{proposition}{\bf\em Proposition}
\newtheorem{corollary}{\bf\em Corollary}
\newtheorem{example}{\bf\em Example}

\newcommand{\Z}{{\mathbb{Z}}}
\newcommand{\C}{{\mathcal{C}}}

\begin{document}
\title{Classification of the $Z_2Z_4$-Linear Hadamard Codes\\ and Their Automorphism Groups%
\thanks{The work of the first author
has been partially supported
by the Russian Foundation
    for Basic Research
      under Grant 13-01-00463-a.
The work of the second author
has been partially supported
by the Spanish MICINN
under Grants TIN2010-17358
         and TIN2013-40524-P
and by the Catalan AGAUR
  under Grant 2014SGR-691.
The material in this paper was presented in part at the 4th International Castle Meeting on Coding Theory and Applications, Palmela, Portugal, 2014 \cite{KroVil:2014castle}.}%
}
\author{Denis S. Krotov%
\thanks{D. Krotov is with the Sobolev Institute of Mathematics, Novosibirsk
630090, Russia, and also with the Novosibirsk State University, Novosibirsk
630090, Russia (e-mail: krotov@math.nsc.ru).}
and Merc\`e Villanueva%
\thanks{M. Villanueva is with the Department of Information and Communications
Engineering, Universitat Aut\`{o}noma de Barcelona, 08193-Bellaterra, Spain
\mbox{(e-mail: merce.villanueva@uab.cat)}.}%
\thanks{This is the peer reviewed version of the paper published in the IEEE Transactions on Information Theory, vol. 61, no. 2, pp. 887--894, 2015, DOI \href{http://doi.org/10.1109/TIT.2014.2379644}{10.1109/TIT.2014.2379644}.
\copyright\ 2014 IEEE. Personal use is permitted, but republication/redistribution requires IEEE permission.}%
}
\date{}

\maketitle
\thispagestyle{empty}
\begin{abstract}
\boldmath
A $Z_2Z_4$-linear Hadamard code of length $\alpha+2\beta=2^t$ is a binary Hadamard code
which is the Gray map image of a
$Z_2Z_4$-additive code with $\alpha$ binary coordinates and $\beta$ quaternary coordinates.
It is known that there are exactly $\lfloor \frac{t-1}{2} \rfloor$ and $\lfloor\frac{t}{2} \rfloor$
nonequivalent $Z_2Z_4$-linear Hadamard codes of length $2^t$, with $\alpha=0$ and $\alpha \not =0$, respectively,
for all $t\geq 3$. In this paper, it is shown that each $Z_2Z_4$-linear Hadamard  code with $\alpha=0$ is equivalent to a $Z_2Z_4$-linear Hadamard  code with $\alpha \not =0$,
so there are only $\lfloor\frac{t}{2} \rfloor$ nonequivalent $Z_2Z_4$-linear Hadamard codes of length $2^t$.
Moreover, the order of the monomial automorphism group for
the $Z_2Z_4$-additive Hadamard codes and the permutation automorphism group of
the corresponding $Z_2Z_4$-linear Hadamard codes are given.
\end{abstract}

\baselineskip=0.96\normalbaselineskip
\def\0{\mathbf{0}}
\def\1{\mathbf{1}}
\def\2{\mathbf{2}}
\def\3{\mathbf{3}}
\section{Introduction}\setstretch{1.06}

\newcommand\doT[1]{\accentset{\mbox{$.$}}#1}

Let $\Z_2$ and $\Z_4$ be the rings of integers modulo 2 and modulo 4,
respectively. Let $\Z_2^n$ be the set of all binary vectors of
length $n$ and let $\Z_4^n$ be the set of all quaternary vectors of
length $n$. For a vector $x=(x_1, \ldots,x_{n})\in \Z_2^n$ and a set $I \subseteq \{1, \ldots,n\}$,
we denote by $x|_{I}$ the vector $x$ restricted to the coordinates in~$I$.

Any nonempty subset $C$ of $\Z_2^n$ is a binary code
and a subgroup of $\Z_2^n$ is called a \emph{binary linear code}.
Similarly, any nonempty subset $\C$ of $\Z_4^n$ is a quaternary
code and a subgroup of $\Z_4^n$ is called a {\it quaternary linear code}.
Let $\C$ be a quaternary linear code.
Since $\C$ is a subgroup of  $\Z_4^n$,
it is isomorphic to an Abelian group
$\Z_2^{\gamma}\times \Z_4^{\delta}$,
for some $\gamma$ and $\delta$, and we say that $\C$ is of type
$2^\gamma4^\delta$ as a group. Quaternary codes can be seen as binary
codes under the usual Gray map defined as $\varphi(0)=(0,0),\ \varphi(1)=(0,1), \ \varphi(2)=(1,1),\ \varphi(3)=(1,0)$
in each coordinate. If $\C$ is a quaternary  linear code, then the binary code $C=\varphi(\C)$ is called a {\em $\Z_4$-linear code}.

Additive codes were first defined by Delsarte in $1973$ as
subgroups of the underlying Abelian group in a translation association scheme \cite{Delsarte:1973,DelsLev:98}.
In the special case of a binary Hamming scheme, that is, when the underlying Abelian group is
of order $2^n$, the additive codes coincide with the codes that are subgroups of
$\Z_2^\alpha\times\Z_4^\beta$. In order to distinguish them from additive
codes over finite fields \cite{Bier}, they are called {\em $\Z_2\Z_4$-additive codes} \cite{BFPRV2010}.
Since $\Z_2\Z_4$-additive codes are subgroups of $\Z_2^{\alpha} \times \Z_4^{\beta}$,
they can be seen as a generalization of binary (when $\beta=0$) and quaternary (when $\alpha=0$) linear codes.
As for quaternary linear codes, $\Z_2\Z_4$-additive codes
can also be seen as binary codes by considering the
extension of the usual Gray map: $\Phi:
\Z_2^{\alpha}\times\Z_4^{\beta} \longrightarrow \Z_2^{n}$, where
$n=\alpha+2\beta$, given by
\begin{multline*}
\Phi(x,y) = (x,\varphi(y_1),\ldots,\varphi(y_\beta))\\
 \forall x\in\Z_2^\alpha,\;\forall y=(y_1,\ldots,y_\beta)\in \Z_4^\beta.
\end{multline*}
If $\C$ is a $\Z_2\Z_4$-additive code, $C=\Phi(\C)$ is called a {\em $\Z_2\Z_4$-linear code}.
Moreover, a $\Z_2\Z_4$-additive code $\C$ is also isomorphic to an Abelian group
$\Z_2^{\gamma}\times \Z_4^{\delta}$,
and we say that $\C$
(or equivalently the corresponding $\Z_2\Z_4$-linear code $C=\Phi(\C)$) is of type
$(\alpha,\beta;\gamma,\delta)$.

Let $S_n$ be the symmetric group of permutations on the set $\{1,\ldots,n\}$, and
let $\mathrm{id} \in S_n$ be the identity permutation.
The group operation in $S_n$ is the function composition.
The composition
$\sigma_1 
\sigma_2$
maps any element $x$ to $\sigma_1(\sigma_2(x))$.
A $\sigma \in S_n$
acts linearly on words of $\Z_2^n$ or $\Z_4^n$ by permuting the coordinates,
$\sigma((c_1,\ldots,c_n))=(c_{\sigma^{-1}(1)},\ldots, c_{\sigma^{-1}(n)})$.

Let $\C$ be a $\Z_2\Z_4$-additive code of type
$(\alpha,\beta;\gamma,\delta)$. We can assign a
permutation $\pi_x \in S_n$ to each codeword $x=( x'_1,\ldots, x'_\alpha, x_1,\ldots,x_{2\beta}) \in C=\Phi(\C)$, such
that $\pi_x=\pi_{12} 
\pi_{34} 
\cdots 
\pi_{2\beta-1\;2\beta}$,
where
\vspace{-0,15truecm}$$
\pi_{ij} = \left\{ \begin{array}{ll}
           \mathrm{id} & \mbox{if $(x_i,x_j)=(0,0)$ or $(1,1)$} \\
           (i,j) & \mbox{otherwise}.
\end{array}
\right.
$$
Given two codewords of $C$, $x=(x',x_1,\ldots,x_{2\beta})$ and $y=(y',y_1,\ldots,y_{2\beta})$,
define $x\star y=x+\pi_x(y)$.
Then, we have that $(C,\star)$ is an Abelian group \cite{RifPuj:1997},
which is isomorphic to $(\C,+)$, since
\begin{IEEEeqnarray*}{rCl}
 x\star y&=&\big(x'+y',\, \varphi(\varphi^{-1}(x_1,x_2)+\varphi^{-1}(y_1,y_2)),
 \\ &&\ \ \ \ \ \ \ldots,\,
\varphi(\varphi^{-1}(x_{2\beta-1},x_{2\beta})+\varphi^{-1}(y_{2\beta-1},y_{2\beta}))\big)
\\
&=&\Phi\big(\Phi^{-1}(x)+\Phi^{-1}(y)\big).
\end{IEEEeqnarray*}

\setstretch{1.05}
There are $\Z_2\Z_4$-linear codes in several important classes of binary codes.
For example, $\Z_2\Z_4$-linear perfect single error-correcting codes (or 1-perfect codes)
are found in \cite{RifPuj:1997} and fully characterized in \cite{BorRif:1999}.
Also, in subsequent papers \cite{Kro:2000:Z4_Perf,BorPheRif:2003, Kro:2001:Z4_Had_Perf, PheRif:2002, PheRifVil:2006},
$\Z_2\Z_4$-linear extended perfect and Hadamard codes are studied and
classified independently for $\alpha=0$ and $\alpha \not =0$.
Finally, in \cite{PujRifSol:2009,Sol:2007RM,PePuVi11},
$\Z_2\Z_4$-linear Reed--Muller codes are also studied.
Note that $\Z_2\Z_4$-linear codes have allowed to classify more binary nonlinear codes,
giving them a structure as $\Z_2\Z_4$-additive codes.

A {\it (binary) Hadamard code} of length $n$ is a binary code with $2n$ codewords and minimum distance $n/2$
\cite{MWS}. The $\Z_2\Z_4$-additive codes that, under the Gray map, give a
Hadamard code are called {\it $\Z_2\Z_4$-additive Hadamard codes} and
the corresponding $\Z_2\Z_4$-linear codes are called {\it  $\Z_2\Z_4$-linear  Hadamard
codes}, or just {\it $\Z_4$-linear Hadamard  codes} when $\alpha=0$.
The classification of  $\Z_2\Z_4$-linear Hadamard codes is given by the following results.
For any integer $t\geq 3$ and each
$\delta\in\{1,\ldots,\lfloor (t+1)/2 \rfloor\}$, there is a unique
(up to equivalence)
$\Z_4$-linear Hadamard code of type $(0,2^{t-1};t+1-2\delta,\delta)$, and all these codes are
pairwise nonequivalent, except for $\delta=1$ and $\delta=2$,
where the codes are equivalent to the linear Hadamard code, that is, the dual of the extended Hamming code \cite{Kro:2001:Z4_Had_Perf}.
Therefore, the number of nonequivalent $\Z_4$-linear
 Hadamard  codes of length $2^t$ is $\lfloor
\frac{t-1}{2}\rfloor$ for all $t\geq 3$.
On the other hand, for any integer $t\geq 3$ and each
$\delta\in\{0,\ldots,\lfloor t/2 \rfloor\}$,  there is a unique (up to equivalence)
$\Z_2\Z_4$-linear Hadamard code of type $(2^{t-\delta},2^{t-1}-2^{t-\delta-1};t+1-2\delta,\delta)$.
All these codes are pairwise nonequivalent, except for $\delta=0$ and $\delta=1$, where the
codes are equivalent to the linear Hadamard code \cite{BorPheRif:2003}. Therefore,
the number of nonequivalent $\Z_2\Z_4$-linear Hadamard  codes
of length $2^t$  with $\alpha \neq 0$ is $\lfloor t/2 \rfloor$ for all $t\geq 3$.
In this paper, we show that any $\Z_4$-linear Hadamard code is equivalent to a
$\Z_2\Z_4$-linear Hadamard code with $\alpha \ne 0$.

Two structural properties of binary codes are the rank and
the dimension of the kernel. The \emph{rank} of a code $C$ is simply the
dimension of the linear span, $\langle C \rangle$,  of $C$.
The \emph{kernel} of a code $C$ is defined as
$\mathrm{Ker}(C)=\{x\in \Z_2^n \,|\, x+C=C \}$ \cite{BGH83}. If the all-zero vector belongs to $C$,
$\mathrm{Ker}(C)$ is a linear subcode of $C$. In general, $C$ can be written
as the union of cosets of $\mathrm{Ker}(C)$, and $\mathrm{Ker}(C)$ is the largest 
linear code for which
this is true \cite{BGH83}.
The $\Z_2\Z_4$-linear Hadamard codes
can also be classified using either the rank or the dimension of the kernel,
as it is proven in \cite{Kro:2001:Z4_Had_Perf,PheRifVil:2006}, where these parameters are computed.

Two $\Z_2\Z_4$-additive codes $\C_1$ and $\C_2$ both of type
$(\alpha,\beta;\gamma,\delta)$ are said to be {\it
monomially equivalent}, if one can be obtained from the other by permuting
the coordinates and (if necessary) changing the signs of certain
$\Z_4$ coordinates. Two $\Z_2\Z_4$-additive or $\Z_2\Z_4$-linear codes are said to be {\it
permutation equivalent} if they differ only by a permutation of
coordinates.
The {\it monomial automorphism group} of a
$\Z_2\Z_4$-additive code $\C$, denoted by  $\mathrm{MAut}(\C)$,
is the group generated by all permutations and sign-changes
of the $\Z_4$ coordinates that preserve the set of codewords of $\C$, while
the {\it permutation automorphism group} of $\C$ or $C=\Phi(\C)$, denoted by $\mathrm{PAut}(\C)$ or $\mathrm{Aut}(C)$,
respectively, is the group generated by all permutations that preserve the set of codewords \cite{HuffmanPless}.

The permutation automorphism group of a binary code is also an invariant, so it can help in the classification of some families of binary codes.
Moreover, the automorphism group can also be used in decoding algorithms and to describe some other properties like the weight distribution \cite[Ch. 8, \S 5 and Ch. 16, \S 9]{MWS}.
The permutation automorphism group of $\Z_2\Z_4$-linear (extended) $1$-perfect codes has been studied in~\cite{PheRif:2002,Kro:2011castle}.
The permutation automorphism group of (nonlinear) binary 1-perfect codes has also been studied before,
obtaining some partial results \cite{HedPasWes:2009,Hed:2005,AvgSolHed:2005,FePhVi:2011}.
For $\Z_2\Z_4$-additive Hadamard codes with $\alpha=0$, the permutation automorphism group was characterized in \cite{PePuVi14},
while the monomial automorphism group as well as the permutation automorphism group of the corresponding $\Z_2\Z_4$-linear codes  had not been studied yet. In this paper, we completely solve this problem giving the
structure and order of these two automorphism groups for~any~$\alpha$.

The structure of the current paper is as follows.
In Section~\ref{s:class}, we suggest a new representation
for the $\Z_2\Z_4$-linear Hadamard codes,
separately for the cases $\alpha =0$ and $\alpha\ne 0$,
and prove that the $\Z_2\Z_4$-linear Hadamard codes of types
$(0,n/2;\gamma,\delta)$
and $(\alpha,(n-\alpha)/2;\gamma+2,\delta-1)$
with $\alpha=n/2^{\delta-1}\ne 0$ are permutation equivalent.
In Sections~\ref{s:mautZ4} and~\ref{s:mautZ2Z4},
we establish the structure and the order
of the monomial automorphism group of the
$\Z_2\Z_4$-additive Hadamard codes,
for $\alpha=0$ and  $\alpha\ne 0$, respectively.
In Sections~\ref{s:ub} and~\ref{s:paut},
we establish the structure and the order
of the permutation automorphism group of the
$\Z_2\Z_4$-linear Hadamard codes.
The main results of the paper are
Theorem~\ref{th:equiv} on the equivalence,
Theorems~\ref{th:MAutZ4} and~\ref{th:MAutZ2Z4}
on the monomial automorphism group
of the $\Z_2\Z_4$-additive Hadamard codes,
Theorems~\ref{th:d3} and~\ref{th:d3-2}
on the permutation automorphism group
of the corresponding binary codes,
and a representation of $\Z_2\Z_4$-linear Hadamard codes
in Section~\ref{s:class}.
\section[Classification of Z2Z4-linear Hadamard codes]{Classification of $\Z_2\Z_4$-linear Hadamard codes}\label{s:class}

In \cite{Kro:2001:Z4_Had_Perf} and \cite{BorPheRif:2003},
$\Z_2\Z_4$-linear Hadamard  codes are classified independently for $\alpha=0$ and $\alpha \not =0$.
In this section, we show that each $\Z_2\Z_4$-linear Hadamard code with $\alpha=0$
is equivalent to a $\Z_2\Z_4$-linear Hadamard code with $\alpha \not =0$,
so there are only $\lfloor\frac{t}{2} \rfloor$ nonequivalent $\Z_2\Z_4$-linear Hadamard codes of length $2^t$.

\medskip
We say that a function $f$
from $\Z_2^{i} \times \Z_4^{j}$
to $\Z_2^{s} \times \Z_4^{t}$
is \emph{affine}
if $f(\overline 0)-f(x)-f(y)+f(x+y)=\overline 0$
for every $x$ and $y$
from $\Z_2^{i} \times \Z_4^{j}$
(here and in what follows,
$\overline 0$
denotes the all-zero vector).
Equivalently, $f(\cdot)-f(\overline 0)$
is a linear function,
i.e., a group homomorphism.
Let $\mathcal{B}=\mathcal{B}_{{\gamma},\doT\delta}$ be the set of all affine functions from $\Z_2^{\gamma} \times \Z_4^{\doT\delta}$ to $\Z_4$.
These $\Z_4$-valued functions on $\Z_2^{\gamma} \times \Z_4^{\doT\delta}$
can be considered as words of length $2^{\gamma+2\doT\delta}$ over $\Z_4$.
Denote $B_{{\gamma},\doT\delta}= \{ x : \Z_2^{\gamma} \times \Z_4^{\doT\delta} \to \Z_2^2 \, | \, x(\cdot)=\varphi(f(\cdot)) \mbox{ for some }  f\in \mathcal{B} \}$.

\begin{lemma}\it
 $B_{{\gamma},\doT\delta}$ is a $\Z_4$-linear Hadamard code
 of length $n=2^{\gamma+2\doT\delta+1}$ and type $(0,n/2;\gamma,\delta)$,
 where $\delta=\doT\delta +1$.
\end{lemma}

\begin{IEEEproof}
There are $4\cdot 2^{\gamma} \cdot 4^{\doT\delta}=2n$ affine functions in $\mathcal{B}$.
The set $\mathcal{B}$ is closed under the addition over $\Z_4$; so after applying the Gray map, $B_{{\gamma},\doT\delta}$
is a $\Z_4$-linear code of length $2^{\gamma} \cdot 4^{\doT\delta} \cdot 2=n$. Clearly, the minimum Hamming
distance is $n/2$.
\end{IEEEproof}

Define the function $\varphi^+:\Z_4\to\{0,1\}$ by $\varphi^+(0)=\varphi^+(3)=0$, $\varphi^+(1)=\varphi^+(2)=1$.
Again, the $\Z_2$-valued or $\Z_4$-valued functions on $\Z_2^{\doT\gamma} \times \Z_4^\delta$
can be considered as words of length $2^{{\doT\gamma}+2\delta}$ over $\Z_2$ or $\Z_4$, respectively.
Let $\mathcal{A}$ be the set of all affine functions $f$ from $\Z_2^{\doT\gamma} \times \Z_4^\delta$ to $\Z_4$ that map the all-zero vector to $0$ or $2$: $f(\overline{0})\in\{0,2\}$.
Denote $C_{{\doT\gamma},\delta}= \{h:\Z_2^{\doT\gamma} \times \Z_4^\delta \to \Z_2 \,|\, h(\cdot)=\varphi^+(f(\cdot)) \mbox{ for some } f\in \mathcal{A}\}$.

\begin{lemma}\it\label{l:constr}
 $C_{{\doT\gamma},\delta}$ is a $\Z_2\Z_4$-linear Hadamard code of length $n=2^{{\doT\gamma}+2\delta}$ and type $(\alpha,\beta; \gamma,\delta)$,
where $\gamma=\doT\gamma+1$, $\alpha=2^{{\doT\gamma}+\delta}$ corresponding to the elements of order at most 2 of $\Z_2^{\doT\gamma} \times \Z_4^\delta$,
and $\beta=2^{{\doT\gamma}+\delta-1}(2^{\delta}-1)$ corresponding to the pairs of opposite elements of order $4$.
\end{lemma}

\begin{IEEEproof}
There are $2\cdot 2^{\doT\gamma} \cdot 4^\delta=2n$ affine
functions in $\mathcal{A}$. The set $\mathcal{A}$ is closed under the addition over $\Z_4$;
so the Gray map image $A=\Phi(\mathcal{A})$ can also be considered as
a $\Z_2\Z_4$-linear code with $2^{{\doT\gamma}+\delta+1}$ coordinates over $\Z_2$, which correspond
to the elements of order at most $2$ of $\Z_2^{\doT\gamma} \times \Z_4^\delta$.

Now, we will see that the code $A$ can be obtained from $C_{{\doT\gamma},\delta}$ by repeating twice every coordinate.
That is, strictly speaking, $A$ is permutation equivalent to $\{(h,h) \,|\, h\in C_{{\doT\gamma},\delta}\}$.
Indeed, given $v\in \Z_2^{\doT\gamma} \times \Z_4^\delta$ of order $4$
and an affine function $f\in \mathcal{A}$, the values $\varphi^+(f(v))$ and $\varphi^+(f(-v))$
of the corresponding codeword of $C_{{\doT\gamma},\delta}$ each occurs both in $\varphi(f(v))$ and $\varphi(f(-v))$.
If the order of $v\in \Z_2^{\doT\gamma} \times \Z_4^\delta$ is $2$ or less, then $\varphi^+(f(v))$ is duplicated in $\varphi(f(v))$.

Finally, it is easy to check that the minimum Lee distance for the set of affine functions $\mathcal{A}$ is $n=2^{\doT\gamma+2\delta}$;
so the minimum Hamming distance of $C_{{\doT\gamma},\delta}$ is the half of this value, $n/2$.
\end{IEEEproof}

\begin{lemma}\it\label{l:13}
Let
$f:\Z_2^{\doT\gamma} \times \Z_4^\delta\to \Z_4$
be an affine function.
Then $h(\cdot) = \varphi^+(f(\cdot))$
belongs to $C_{{\doT\gamma},\delta}$.
\end{lemma}

\begin{IEEEproof}
 In the case that $f(\overline{0})\in \{0,2\}$,
 $C_{{\doT\gamma},\delta}$ contains $h$ by definition.
 On the other hand,
 if $f(\overline{0}) \in \{1,3\}$,
 we will use that
 $\varphi^+(l)=\varphi^+(3-l)$
    for $l\in \Z_4$.
 Then,
 $h(\cdot)=\varphi^+(f(\cdot))=\varphi^+(3-f(\cdot))$.
 Since $3-f(\cdot)$
 is an affine function
 and $3-f(\overline{0})\in \{0,2\}$,
 we obtain that
 $h\in C_{{\doT\gamma},\delta}$.
\end{IEEEproof}

\begin{theorem}\it\label{th:equiv}
The $\Z_4$-linear Hadamard code
$B_{{\gamma},\doT\delta}$ of length $n$
and type $(0,n/2;\gamma,\delta)$
is permutation equivalent
 to the $\Z_2\Z_4$-linear Hadamard code
 $C_{\gamma+1,\doT\delta}$
 of type $(\alpha,(n-\alpha)/2;\gamma+2,\delta-1)$
 with $\alpha=n/2^{\delta-1} \not =0$.
\end{theorem}

\begin{IEEEproof}
Consider a function $f$ in $\mathcal{B}$ and the related function $g(v,e)= f(v) + 2ef(\overline{0})$,
where $v\in \Z_2^{\gamma} \times \Z_4^{\doT\delta}$ and $e\in \Z_2$.
We can see that
\begin{IEEEeqnarray*}{rCl}
 \varphi(f(v)) &=& \left(\varphi^+(g(-v,1)),\varphi^+(g(v,0))\right),\\
 \varphi(f(-v)) &=& \left(\varphi^+(g(v,1)),\varphi^+(g(-v,0))\right).
\end{IEEEeqnarray*}
In order to check these equalities, it is convenient to represent $f(v)$ as $f_0(v) + f(\overline{0})$, where
$f_0 : \Z_2^{\gamma} \times \Z_4^{\doT\delta}\to \Z_4$ is a group homomorphism (in particular, $f_0(-v)=-f_0(v)$).

Since $g$ is an affine function from $v\in \Z_2^{\gamma+1} \times \Z_4^{\doT\delta}$ to $\Z_4$,
we can deduce from Lemma~\ref{l:13} that there is a fixed coordinate permutation that sends
every codeword of $B_{{\gamma},\doT\delta}$ to a codeword of $C_{{\gamma+1},\doT\delta}$.
\end{IEEEproof}

\begin{corollary}\it
There are exactly $\lfloor\frac{t}{2} \rfloor$ nonequivalent $\Z_2\Z_4$-linear Hadamard codes of length $2^t$.
\end{corollary}

\section[The monomial automorphism group of the Z2Z4-additive Hadamard codes, alpha=0]{The monomial automorphism group of the $\Z_2\Z_4$-additive Hadamard codes, $\alpha =0$}\label{s:mautZ4}
\setstretch{1.10}
In this and next sections, we establish
the structure and the order of
the monomial automorphism group for
the $\Z_2\Z_4$-additive Hadamard codes.
We start with the $\Z_2\Z_4$-additive Hadamard codes
with $\alpha = 0$.

Recall 
that the $\Z_2\Z_4$-additive Hadamard code
of type $(0,2^{\gamma+2\doT\delta};\gamma,\doT\delta+1)$ can be considered as
the set $\mathcal{B}=\mathcal{B}_{{\gamma},\doT\delta}$ of all affine functions from $\Z_2^{\gamma} \times \Z_4^{\doT\delta}$ to $\Z_4$.
As the elements of $\Z_2^{\gamma} \times \Z_4^{\doT\delta}$
play the role of coordinates,
the coordinate permutations are the permutations of
$\Z_2^{\gamma} \times \Z_4^{\doT\delta}$,
and the action of such permutation $\sigma$
on a function $f$ on $\Z_2^{\gamma} \times \Z_4^{\doT\delta}$
can be expressed as
$\sigma f(v) = f(\sigma^{-1}(v))$.
The action of a sign change $\tau$ can be expressed as
$\tau f(v) = f(v)r(v)$,
where $r$ is the corresponding function
from $\Z_2^{\gamma} \times \Z_4^{\doT\delta}$ to $\{1,3\}$.

Assume we have an affine function $r$ from $\Z_2^{\gamma} \times \Z_4^{\doT\delta}$ to $\{1,3\}\subset\Z_4$;
note that $r(v)=r(-v)$ for all $v$.
Let $\mathrm{supp}(r)$ be the set of elements of $\Z_2^{\gamma} \times \Z_4^{\doT\delta}$
whose image by $r$ is 3.
By $\tau_r$, we denote
the transformation
that changes the sign of the value of a function $f\in \mathcal{B}_{{\gamma},\doT\delta}$
at the coordinates given by the elements of $\mathrm{supp}(r)$: $\tau_r f(v) = f(v)r(v)$.
By $\rho_r$, we denote the permutation
of $\Z_2^{\gamma} \times \Z_4^{\doT\delta}$
that changes the sign of the elements of $\mathrm{supp}(r)$:
$\rho_r f(v) = f(r(v)v)$.

\begin{theorem}\it\label{th:MAutZ4}
The monomial automorphism group $\mathrm{MAut}(\mathcal{B}_{{\gamma},\doT\delta})$
of the $\Z_2\Z_4$-additive Hadamard code $\mathcal{B}_{{\gamma},\doT\delta}$ of type
$(0,2^{\gamma+2\doT\delta};\gamma,\doT\delta+1)$
consists of all transformations $\tau_r  \rho_r  \sigma$
where $\sigma$ is an affine permutation of $\Z_2^{\gamma} \times \Z_4^{\doT\delta}$
and $r$ is an affine function from $\Z_2^{\gamma} \times \Z_4^{\doT\delta}$ to $\{1,3\}$
($\tau_r$ and $\rho_r$ are the sign change
and the permutation of $\Z_2^{\gamma} \times \Z_4^{\doT\delta}$
assigned to $r$ as above).
The cardinality of this group satisfies
\begin{multline*}
|\mathrm{MAut}(\mathcal{B}_{{\gamma},\doT\delta})| \\=
2^{\frac{1}{2}{\gamma}^2+\frac{3}{2}{\gamma} +2{\gamma}\doT\delta + \frac{3}{2}\doT\delta^2 + \frac52\doT\delta+1}\prod_{i=1}^{{\gamma}}(2^i-1)\prod_{j=1}^{\doT\delta}(2^j-1).
\end{multline*}
\end{theorem}
\begin{IEEEproof}
The statement of the theorem is straightforward from the following three claims.

(i) \emph{Every transformation from the statement of the theorem belongs to $\mathrm{MAut}(\mathcal{B})$}.
Obviously, $\sigma$ sends $\mathcal{B}$ to $\mathcal{B}$, as the composition of an affine permutation and an affine function is an affine function, too. It remains to check that $\tau_r  \rho_r \in \mathrm{MAut}(\mathcal{B})$.
Consider an affine function $g(\cdot) = l(\cdot)+c(\cdot)$, where $l$ is a linear function and $c$ is a constant function,   $c(\cdot)\equiv a$. Since $l(-v)=-l(v)$ and $r(v)=r(-v)$, we have $\tau_r ( \rho_r (l))= l$.
Further, $\rho_r c (\cdot)= c(\cdot)$ and $\tau_r c(\cdot)= a r(\cdot)$.
Therefore, we see that  $\tau_r  \rho_r (g)$ is an affine function and claim (i) is proved.

(ii) \emph{If $\tau  \psi\in \mathrm{MAut}(\mathcal{B})$, where $\psi$ is a permutation of $\Z_2^{\gamma} \times \Z_4^{\doT\delta}$ and $\tau$ is a sign change, then $\tau = \tau_r$ for some affine function $r$}.
Indeed, consider the constant function $g\equiv 1$, which obviously belongs to $\mathcal{B}$.
Then $\psi(g) = g$. Denote $r(\cdot) = \tau g (\cdot)$.
Since $\tau(g)=\tau(\psi(g))$ is affine, we find that $r$ is also affine. It is straightforward that $\tau = \tau_r$.

(iii)
\emph{Every permutation $\sigma$ of
$\Z_2^{\gamma} \times \Z_4^{\doT\delta}$
from $\mathrm{MAut}(\mathcal{B})$ is affine}.
Let  $\sigma'_1$, \ldots, $\sigma'_\gamma:\Z_2^{\gamma} \times \Z_4^{\doT\delta} \to \Z_2$
and $\sigma''_1$, \ldots, $\sigma''_{\doT\delta}:\Z_2^{\gamma} \times \Z_4^{\doT\delta} \to \Z_4$
be the components of the vector-function $\sigma^{-1}$. That is,
$$\sigma^{-1}(v) = \big(\sigma'_1(v),\ldots,\sigma'_{\gamma}(v),\sigma''_1(v),\ldots,\sigma''_{\doT\delta}(v)\big).$$
Now, consider the function $f_i:\Z_2^{\gamma} \times \Z_4^{\doT\delta} \to \Z_4$ defined as
$f_i(x_1,\ldots,x_{\gamma},y_1,\ldots, y_{\doT\delta})=2x_i$, $i\in\{1,\ldots,\gamma\}$.
This function is affine and,
as $\sigma \in \mathrm{MAut}(\mathcal{B})$,
the function $\sigma f_i$ is also affine.
We have $\sigma f_i (v) = f_i(\sigma^{-1}(v)) = 2\sigma'_i(v)$,
which means that $\sigma'_i(v)$ is an affine $\Z_2$-valued function.
Similarly, considering
$g_j(x_1,\ldots,x_{\gamma},y_1,\ldots,y_{\doT\delta})=y_j$,
we see that $\sigma''_j(v)$ is affine, $j\in\{1,\ldots,\doT\delta\}$.
Since all the components are affine,
we conclude that $\sigma^{-1}$ and $\sigma$ are affine.

Now, assume we have a transformation
$\xi = \tau  \psi$ from $\mathrm{MAut}(\mathcal{B})$,
where $\psi$ is a permutation of
$\Z_2^{\gamma} \times \Z_4^{\doT\delta}$
and $\tau$ is a sign change.
By claim (ii), we have $\tau  = \tau_r$ for some affine $\{1,3\}$-valued function $r$.
By claim (i), the permutation $\sigma = \rho_r^{-1} \psi = \rho_r^{-1} \tau_r^{-1} \xi$
belongs to $\mathrm{MAut}(\mathcal{B})$,
and by claim (iii) it is affine. Then, $\xi = \tau_r \rho_r \sigma$.

Let us calculate the cardinality of the group.
An affine permutation is uniquely represented
as the composition of a linear permutation and a translation.
There are
$T=|\Z_2^{\gamma} \times \Z_4^{\doT\delta}|
=2^{\gamma+2\doT\delta}$
different translations.
To choose a linear permutation of
$\Z_2^{\gamma} \times \Z_4^{\doT\delta}$,
one should choose the images
$y_1$, \ldots, $y_{\doT\delta}$, $x_1$, \ldots, $x_{\gamma}$
of the elements $(10...0)$, $(010...0)$, \ldots,
$(0...01)$, respectively, i.e.,
 ${\doT\delta}$ elements of order $4$
and ${\gamma}$ elements of order $2$, spanning
$\Z_2^{\gamma} \times \Z_4^{\doT\delta}$.
The first element $y_1$ can be chosen in
$2^\gamma 4^{\doT\delta} - 2^\gamma 2^{\doT\delta}$ ways.
To choose $y_j$,
one can choose an element of order $4$
in the factor-group
$\Z_2^{\gamma} \times \Z_4^{\doT\delta}
/ \langle y_1,\ldots, y_{j-1} \rangle$
(in $2^\gamma 4^{\doT\delta-j+1} - 2^\gamma 2^{\doT\delta-j+1}$ ways)
and a representative in the chosen coset of
$\langle y_1,\ldots, y_{j-1} \rangle$
(in $4^{j-1}$ ways). Totally, there are
\begin{IEEEeqnarray*}{rCl}
 Y&=&\prod_{j=1}^{\doT\delta}
\big(4^{\doT\delta}2^{\gamma}-
2^{\doT\delta}2^{\gamma}2^{j-1}\big)\\[-1mm]
&=&
2^{\doT\delta^2+\gamma\doT\delta+\frac12\doT\delta(\doT\delta-1)}
\prod_{j=1}^{\doT\delta}(2^j-1)
\end{IEEEeqnarray*}
ways to choose $y_1$, \ldots, $y_{\doT\delta}$.
After choosing $y_1$, \ldots, $y_{\doT\delta}$,
$x_1$, \ldots, $x_{i-1}$,
the element $x_i$
can be chosen arbitrarily from
the elements of order $2$ of
$\Z_2^{\gamma} \times \Z_4^{\doT\delta}
\setminus \langle y_1,\ldots, y_{\doT\delta},
x_1, \ldots, x_{i-1}\rangle$, in
$2^{\doT\delta}2^{\gamma}-2^{\doT\delta}2^{i-1}$
ways. Totally, there are
\begin{IEEEeqnarray*}{rCl}
 X&=&\prod_{i=1}^{\gamma}
\big(2^{\doT\delta}2^{\gamma}-
2^{\doT\delta}2^{i-1}\big)
\\[-1mm]&=&
2^{\doT\delta\gamma+\frac12\gamma(\gamma-1)}
\prod_{i=1}^{\gamma}(2^i-1)\ \ \ 
\end{IEEEeqnarray*}
ways to choose $x_1$, \ldots, $x_{\gamma}$ after $y_1$, \ldots, $y_{\doT\delta}$. An affine function $r$
is the sum of the coordinates with coefficients from $\{0,2\}$ and a constant from $\{1,3\}$;
so there are $F=2^{\gamma+\doT\delta+1}$ such functions.
Finally, $|\mathrm{MAut}(\mathcal{B})|=TYXF$.
\end{IEEEproof}
\pagebreak

\section[The monomial automorphism group of the Z2Z4-additive Hadamard codes, alpha>0]{The monomial automorphism group of the $\Z_2\Z_4$-additive Hadamard codes, $\alpha\ne 0$}\label{s:mautZ2Z4}
\setstretch{1.05}
Recall that the $\Z_2\Z_4$-linear Hadamard code $C_{{\doT\gamma},\delta}$ of type
$(2^{{\doT\gamma}+\delta},\beta; \doT\gamma+1,\delta)$
can be considered as the set of functions
$\varphi^+(f(\cdot))$ where $f$ belongs to the set $\mathcal{A}$ of affine functions
from $\Z_2^{\doT\gamma} \times \Z_4^\delta$ to $\Z_4$ that map the all-zero vector to $0$ or $2$.
Equivalently,
\begin{multline}\label{eq:defA}
\mathcal{A}= \big\{f:\Z_2^{\doT\gamma} \times \Z_4^\delta\to \Z_4 \ \big|\  f\mbox{ is affine and}\\ 
             f(-x)=-f(x) \ \ \forall x \big\}.
\end{multline}

A pair of opposite elements of $\Z_2^{\doT\gamma} \times \Z_4^\delta$, considered
as a pair of coordinates of the code $C_{{\doT\gamma},\delta}$, corresponds
to one $\Z_4$ coordinate of the corresponding
$\Z_2\Z_4$-additive Hadamard code $\mathcal{C}_{{\doT\gamma},\delta}=\Phi^{-1}(C_{{\doT\gamma},\delta})$.
A self-opposite element (i.e., an element of order $1$ or $2$) of $\Z_2^{\doT\gamma} \times \Z_4^\delta$
corresponds to a $\Z_2$ coordinate of $\mathcal{C}_{{\doT\gamma},\delta}$.
Hence, every monomial automorphism of $\mathcal{C}_{{\doT\gamma},\delta}$
corresponds to a permutation automorphism $\sigma$ of $C_{{\doT\gamma},\delta}$ that
preserves the negation, i.e., $\sigma(-x)=-\sigma(x)$.
Denote the set of all such automorphisms by
$\mathrm{Aut^-}(C_{{\doT\gamma},\delta})$; so,
$\mathrm{Aut^-}(C_{{\doT\gamma},\delta})\simeq \mathrm{MAut}(\mathcal{C}_{{\doT\gamma},\delta})$.

We start from considering the automorphisms of $\mathcal{A}$.
\begin{lemma}\it\label{l:AutA}
$\mathrm{PAut}(\mathcal{A})$ is the set of all affine permutations of $\Z_2^{\doT\gamma} \times \Z_4^\delta$ that
preserve the negation.
\end{lemma}
\begin{IEEEproof}
By (\ref{eq:defA}), any affine permutation $\sigma$ of $\Z_2^{\doT\gamma} \times \Z_4^\delta$ that
satisfies $\sigma(-x)=-\sigma(x)$
is an automorphism of $\mathcal{A}$.

Following the arguments of claim (iii) in the proof of Theorem~\ref{th:MAutZ4}, we can see
that any component of $\sigma$ from
$\mathrm{PAut}(\mathcal{A})$ is affine and, moreover, preserves the negation.
Hence, it is true for the whole $\sigma$.
\end{IEEEproof}

\begin{lemma}\it\label{l:AutM} It holds
$\mathrm{Aut^-}(C_{{\doT\gamma},\delta})=\mathrm{PAut}(\mathcal{A})$.
\end{lemma}
\begin{IEEEproof}
It is straightforward from the definition of
$C_{{\doT\gamma},\delta}$ that
$\mathrm{PAut}(\mathcal{A}) \subseteq \mathrm{Aut}(C_{{\doT\gamma},\delta})$.
Since by Lemma~\ref{l:AutA} the permutations from
$\mathrm{PAut}(\mathcal{A})$
preserve the negation, we have
$\mathrm{PAut}(\mathcal{A}) \subseteq \mathrm{Aut^-}(C_{{\doT\gamma},\delta})$.

Given a $\Z_2$-valued function $x$
from $C_{{\doT\gamma},\delta}$,
we can reconstruct the corresponding
$\Z_4$-valued function $f_x$ from $\mathcal{A}$ by the rule
$f_x(v) = \varphi^{-1}(x(v),x(-v))$.
Since an automorphism $\sigma$ from $\mathrm{Aut^-}(C_{{\doT\gamma},\delta})$ preserves the negation, we have
\begin{IEEEeqnarray*}{rCl}
\sigma f_x(v)
&=& f_{x}(\sigma^{-1}(v))
= \varphi^{-1}\big(x(\sigma^{-1}(v)), x(-\sigma^{-1}(v))\big)\\
&=& \varphi^{-1}\big(x(\sigma^{-1}(v)), x(\sigma^{-1}(-v))\big) \\
&=& \varphi^{-1}\big(\sigma x(v),\sigma x(-v)\big)
= f_{\sigma x}(v).
\end{IEEEeqnarray*}
Consequently, such $\sigma$
sends functions from $\mathcal{A}$
to functions from $\mathcal{A}$, which means
$\mathrm{Aut^-}(C_{{\doT\gamma},\delta}) \subseteq \mathrm{PAut}(\mathcal{A})$.
\end{IEEEproof}
The above lemmas and direct calculations,
similar to that in Section~\ref{s:mautZ4},
give the following result.
\begin{theorem}\it\label{th:MAutZ2Z4}
The monomial automorphism group
$\mathrm{MAut}(\mathcal{C}_{{\doT\gamma},\delta})$
of the $\Z_2\Z_4$-additive Hadamard code
$\mathcal{C}_{{\doT\gamma},\delta}$ of type
$(2^{{\doT\gamma}+\delta},\beta; \doT\gamma+1,\delta)$
is isomorphic to the group of affine permutations
of $\Z_2^{\doT\gamma} \times \Z_4^\delta$ that
preserve the negation. The cardinality of this group satisfies
$$ |\mathrm{MAut}(\mathcal{C}_{{\doT\gamma},\delta})| =
2^{\frac{1}{2}{\doT\gamma^2}+\frac{1}{2}{\doT\gamma} +2{\doT\gamma}\delta + \frac{3}{2}\delta^2+\frac{1}{2}\delta}\prod_{i=1}^{{\doT\gamma}}(2^i-1)\prod_{j=1}^{\delta}(2^j-1).
$$
\end{theorem}
\pagebreak

\section{An upper bound on the order of the permutation automorphism group}\label{s:ub}

In this and next sections, we establish
the structure and the order of
the permutation automorphism group for
the $\Z_2\Z_4$-linear Hadamard codes. In the rest of the paper, we consider $\Z_2\Z_4$-linear codes,
which are binary codes, and in the examples, we will use
$\0$, $\1$, $\2$, and $\3$ as short notations for $00$, $01$, $11$, and $10$, respectively,
placed in the coordinates $\alpha+2i-1$, $\alpha+2i$, $i\in\{1,\ldots,\beta\}$.

Let $C=C_{{\doT\gamma},\delta}$ be a $\Z_2\Z_4$-linear Hadamard code of length $n$
and type $(\alpha, \beta; {\doT\gamma}+1, \delta)$
with $\alpha \not = 0$, where $n=2^{{\doT\gamma}+2\delta}$.
The code $C$ is generated by $\gamma={\doT\gamma}+1$ words $y$, $u_1, \ldots, u_{\doT\gamma}$ of order $2$
and $\delta$ words $v_1, \ldots, v_\delta$ of order $4$, where $y$ is the all-one word of length $n$.
Let $G=G_{{\doT\gamma},\delta}$ be the matrix formed by the rows $y$, $u_1, \ldots, u_{\doT\gamma}$, $v_1, \ldots, v_\delta$.
Let $\cal G$ be the matrix where the columns are all $2^{\doT\gamma+\delta}$ binary columns from $\{1\}\times \{0,1\}^{\delta+\doT\gamma}$ and
half of $2^{\doT\gamma+2\delta}-2^{\doT\gamma+\delta}$  quaternary
columns of order $4$ from $\{2\}\times \{0,1,2,3\}^{\delta}\times \{0,2\}^{\doT\gamma}$,
each pair of opposite column vectors being represented by only one column, see e.g. \cite{PheRifVil:2006}.
Then, $G$ can be seen as the matrix $\cal G$ where the quaternary columns are written as pairs of binary columns under the Gray map.
We assume that $\delta \geq 2$, since
for $\delta=0$ and $\delta=1$, $C$ is the binary linear Hadamard code of length $n=2^t$,
so its permutation automorphism group is the general affine group $\mathrm{GA}(t,2)$,
which has order $2^t(2^t-1)(2^t-2)\cdots(2^t-2^{t-1})$ \cite{MWS}.

The kernel of the code $C$, $\mathrm{Ker}(C)$, is generated by the $1+\delta+{\doT\gamma}$ words $y$, $w_1, \ldots, w_\delta$, $u_1, \ldots, u_{\doT\gamma}$, where $w_j=v_j \star v_j$, of order 2 \cite{PheRifVil:2006}. Let $K=K_{{\doT\gamma},\delta}$ be the matrix formed by these rows $y$, $w_1, \ldots, w_\delta$, $u_1, \ldots, u_{\doT\gamma}$. We assume w.l.o.g. that the columns of $K$ are ordered lexicographically.

\begin{lemma}[\cite{PheRifVil:2006}]\it \label{l:ker}
The matrix $K$ is a generator matrix for $\mathrm{Ker}(C)$.
In particular, $\mathrm{Ker}(C)$ consists of all the elements
of the group $(C,\star)$ of order $2$ or less (without $\1$s and $\3$s).
\end{lemma}

\begin{example} \label{ex:1}
Let $C_{1,2}$ be the $\Z_2\Z_4$-linear Hadamard of length $n=32$ and type $(8,12;2,2)$, that is, generated by $\gamma={\doT\gamma}+1=2$
words of order 2 and $\delta=2$ words of order 4. Then,
$$
G_{1,2}=\left(\begin{array}{c@{\ }c@{\ \ }c@{\ }c@{\ \ }c@{\ }c@{\ \ }c@{\ }c}
1111& 1111& \2\2& \2\2& \2\2& \2\2& \2\2& \2\2\\ \hline
0011& 0011& \0\2& \0\2& \1\1& \1\1& \1\1& \1\1\\
0101& 0101& \1\1& \1\1& \0\2& \0\2& \1\3& \1\3\\ \hline
0000& 1111& \0\0& \2\2& \0\0& \2\2& \0\0& \2\2
\end{array} \right)
$$
and
\begin{IEEEeqnarray*}{l}
K_{1,2}=\left(\begin{array}{c@{\ }c@{\ \ }c@{\ }c@{\ \ }c@{\ }c@{\ \ }c@{\ }c}
1111& 1111& \2\2& \2\2& \2\2& \2\2& \2\2& \2\2\\ \hline
0000& 0000& \0\0& \0\0& \2\2& \2\2& \2\2& \2\2\\
0000& 0000& \2\2& \2\2& \0\0& \0\0& \2\2& \2\2\\ \hline
0000& 1111& \0\0& \2\2& \0\0& \2\2& \0\0& \2\2
\end{array}\right) \phantom{\quad\mbox{and}\qquad}
\\
=\left(\begin{array}{c@{\ }c@{\ \ }c@{\ }c@{\ \ }c@{\ }c@{\ \ }c@{\ }c}
1111& 1111& 11\,11& 11\,11& 11\,11& 11\,11& 11\,11& 11\,11\\ \hline
0000& 0000& 00\,00& 00\,00& 11\,11& 11\,11& 11\,11& 11\,11\\
0000& 0000& 11\,11& 11\,11& 00\,00& 00\,00& 11\,11& 11\,11\\ \hline
0000& 1111& 00\,00& 11\,11& 00\,00& 11\,11& 00\,00& 11\,11
\end{array}\right)\!.
\end{IEEEeqnarray*}
\end{example}
\pagebreak
\begin{example} \label{ex:2}
Let $C_{0,3}$ be the $\Z_2\Z_4$-linear Hadamard code of length $n=64$ and type $(8,28;1,3)$, that is,
generated by $\gamma={\doT\gamma}+1=1$
word of order $2$ and $\delta=3$ words of order $4$. Then, $G_{0,3}$ and $K_{0,3}$ are, respectively,
\begin{eqnarray*}
\left(\begin{array}{@{\,}c@{\ }c@{\ }c@{\ }c@{\ }c@{\ }c@{\ }c@{\ }c@{\,}}
11111111&\2\2\2\2&\2\2\2\2&\2\2\2\2&\2\2\2\2&\2\2\2\2&\2\2\2\2&\2\2\2\2\\ \hline
00001111&\0\0\2\2&\0\0\2\2&\0\0\2\2&\1\1\1\1&\1\1\1\1&\1\1\1\1&\1\1\1\1\\
00110011&\0\2\0\2&\1\1\1\1&\1\1\1\1&\0\0\2\2&\0\0\2\2&\1\1\3\3&\1\1\3\3\\
01010101&\1\1\1\1&\0\2\0\2&\1\3\1\3&\0\2\0\2&\1\3\1\3&\0\2\0\2&\1\3\1\3\\ \hline
\end{array} \right), \\
\\
\left(\begin{array}{@{\,}c@{\ }c@{\ }c@{\ }c@{\ }c@{\ }c@{\ }c@{\ }c@{\,}}
11111111&\2\2\2\2&\2\2\2\2&\2\2\2\2&\2\2\2\2&\2\2\2\2&\2\2\2\2&\2\2\2\2\\ \hline
00000000&\0\0\0\0&\0\0\0\0&\0\0\0\0&\2\2\2\2&\2\2\2\2&\2\2\2\2&\2\2\2\2\\
00000000&\0\0\0\0&\2\2\2\2&\2\2\2\2&\0\0\0\0&\0\0\0\0&\2\2\2\2&\2\2\2\2\\
00000000&\2\2\2\2&\0\0\0\0&\2\2\2\2&\0\0\0\0&\2\2\2\2&\0\0\0\0&\2\2\2\2\\ \hline
\end{array}\right).
\end{eqnarray*}
\end{example}

All coordinates are naturally divided into $2^{{\doT\gamma}+\delta}$ groups of size $2^{\delta}$,
which will be referred to as \emph{blocks}, such that the columns of $K$ in a block coincide.
Similarly, by the values in the rows of $K$ given by $w_1, \ldots, w_\delta$, the coordinates
are divided into $2^{\delta}$ groups of size $\alpha=2^{{\doT\gamma}+\delta}$,
which will be referred to as \emph{macroblocks}.

As representatives of the kernel cosets that partition the code $C$,
we can choose the combinations $v(s)=s_1v_1\star  \cdots \star s_\delta v_\delta$,
where $s=(s_1, \ldots,s_\delta) \in \{0,1\}^\delta$, that is, $C=\bigcup_{s \in \{0,1\}^\delta} (\mathrm{Ker}(C)+v(s))$ \cite{FePuVi08rk}.
Let $S=S_{{\doT\gamma},\delta}$ be the matrix consisting
from the vectors $v(s)$, $s\in\{0,1\}^\delta$, as rows.


\begin{lemma}\it\label{l:v(s)}
For any
$s_1, \ldots,s_\delta \in \{0,1\}$,
we have
\begin{equation}\label{eq:star-to-plus}
v(s)=s_1v_1\star \cdots \star s_\delta v_\delta
= \sum_{j=1}^{\delta}s_jv_j
+\!\!\!\sum_{1\le j< j'\le \delta}\!s_j s_{j'} w_j \bullet w_{j'},
\end{equation}
where $\bullet$ is the component-wise product of two vectors.
\end{lemma}

\begin{IEEEproof}
Straightforward using that $v_i \star v_j=v_i+v_j+ w_i \bullet w_j$ \cite{HammonsOth:Z4_linearity,FePuVi08rk}.
\end{IEEEproof}

\begin{example}\label{ex:S}
Considering the $\Z_2\Z_4$-linear Hadamard code $C_{1,2}$ from Example~\ref{ex:1}, 
$S_{1,2}$ is
\begin{IEEEeqnarray*}{c}
S_{1,2}=\left(\!\!\begin{array}{c}
v(00)\\
v(10)\\
v(01)\\
v(11)\\
\end{array}\!\!\right) =
\left(\!
\begin{array}{c@{\ }c@{\ \ }c@{\ }c@{\ \ }c@{\ }c@{\ \ }c@{\ }c}
0000& 0000&  \0\0&  \0\0&  \0\0&  \0\0&  \0\0&  \0\0\\
0011& 0011&  \0\2&  \0\2&  \1\1&  \1\1&  \1\1&  \1\1\\
0101& 0101&  \1\1&  \1\1&  \0\2&  \0\2&  \1\3&  \1\3\\
0110& 0110&  \1\3&  \1\3&  \1\3&  \1\3&  \2\0&  \2\0\\
\end{array}
\!\right)
\\
\ =\left(\!
\begin{array}{c@{\ }c@{\ \ }c@{\ }c@{\ \ }c@{\ }c@{\ \ }c@{\ }c}
0000& 0000& 00\,00& 00\,00& 00\,00& 00\,00& 00\,00& 00\,00\\
0011& 0011& 00\,11& 00\,11& 01\,01& 01\,01& 01\,01& 01\,01\\
0101& 0101& 01\,01& 01\,01& 00\,11& 00\,11& 01\,10& 01\,10\\
0110& 0110& 01\,10& 01\,10& 01\,10& 01\,10& 11\,00& 11\,00\\
\end{array}\!\right)\!.
\end{IEEEeqnarray*}
Similarly, considering the $\Z_2\Z_4$-linear Hadamard code $C_{0,3}$ from Example~\ref{ex:2}, 
$$S_{0,3}\!=\!\left(\begin{array}{@{}c@{}}
v(000)\\
v(100)\\
v(010)\\
v(110)\\
v(001)\\
v(101)\\
v(011)\\
v(111)\\
\end{array}\right) \!=\!
\left(
\begin{array}{@{\,}c@{\,}c@{\,}c@{\,}c@{\,}c@{\,}c@{\,}c@{\,}c}
00000000&\0\0\0\0&\0\0\0\0&\0\0\0\0&\ldots&\0\0\0\0\\
00001111&\0\0\2\2&\0\0\2\2&\0\0\2\2&\ldots&\1\1\1\1\\
00110011&\0\2\0\2&\1\1\1\1&\1\1\1\1&\ldots&\1\1\3\3\\
00111100&\0\2\2\0&\1\1\3\3&\1\1\3\3&\ldots&\2\2\0\0\\
01010101&\1\1\1\1&\0\2\0\2&\1\3\1\3&\ldots&\1\3\1\3\\
01011010&\1\1\3\3&\0\2\2\0&\1\3\3\1&\ldots&\2\0\2\0\\
01100110&\1\3\1\3&\1\3\1\3&\2\0\2\0&\ldots&\2\0\0\2\\
01101001&\1\3\3\1&\1\3\3\1&\2\0\0\2&\ldots&\3\1\1\3\\
\end{array}
\right)\!.$$
\end{example}
\pagebreak\setstretch{1.02}
\begin{lemma}\it\label{l:fix-block}
 Any automorphism of $C$ stabilizes the partitions of the coordinates into blocks
and macroblocks.
\end{lemma}
\begin{IEEEproof}
Since any automorphism of the code $C$ is an automorphism of its kernel $\mathrm{Ker}(C)$ \cite{PheRif:2002},
the image of any block is obviously a block.

In order to find an invariant for the macroblocks, we use the structure of the linear span of $C$.
As was shown in \cite{PheRifVil:2006}, the linear span is generated
by the words
\begin{equation}\label{eq:basis}
\begin{array}{ll}
y;\\
w_j=v_j\star v_j, & j\in \{1,\ldots,\delta\};\\
w_j \bullet w_{j'},     & 1\le j <j' \le \delta;\\
u_i,              & i\in \{1,\ldots,\doT\gamma\};\\
v_j,              & j\in \{1,\ldots,\delta\}.
\end{array}
\end{equation}

It is straightforward to see the following fact:

(i) \emph{The matrix formed from
$w_1$, \ldots, $w_\delta$,
$v_1$, \ldots, $v_\delta$,
$u_1$, \ldots, $u_{\doT\gamma}$
as rows consists of
all different binary columns
of height $\doT\gamma + 2\delta$}.

Then, we observe the next two facts:

(ii) \emph{Every word of the linear span
       whose weight is different from $n/2$
       has the same value in the coordinates
       of a fixed macroblock.}
       Indeed,
       if a linear combination $x$
       of the basis words (\ref{eq:basis})
       has different values
       in the same macroblock,
       then either some $u_i$ or some $v_j$
       is involved to generate $x$.
       Assume $z$, which is $u_i$ or $v_j$,
       is involved.
       As follows from claim (i),
       for every coordinate,
       there is another one
       such that the values of $z$
       are different in these two coordinates,
       while the values of any other
       basis word from (\ref{eq:basis}) coincide.
       Hence, $x$ has different values
       in each such pair of coordinates,
       which means that its weight is $n/2$.

(iii) \emph{For every two different macroblocks $I$ and $I'$,
 the linear span contains a word
 of weight $n/4$ that has different
 values in the coordinates from these two macroblocks.}
 Indeed, by the definition of the macroblocks,
 there is $w_j$ that has different values
 in the coordinates of $I$ and $I'$.
 Then, for any other $w_{j'}$
 (recall that $\delta \geq 2$), we have
 $w_j = w_j \bullet w_{j'} + w_j \bullet (w_{j'}+y)$.
 Hence,
 at least one of two words $w_j \bullet w_{j'}$
 and $w_j \bullet (w_{j'}+y)$ of weight $n/4$
 differs in the coordinates of $I$ and $I'$.

As follows from claims (ii) and (iii),
there is an invariant
that indicates whether two coordinates
belong to the same macroblock or not:
coordinates $k$ and $k'$
are in different macroblocks
if and only if
the liner span of $C$
contains a word $x=(x_1,\ldots,x_n)$
of weight less than $n/2$
such that $x_k \ne x_{k'}$.
As an automorphism of a code
is obviously an automorphism of its linear span,
the partition of the coordinates
into macroblocks is stabilized.
\end{IEEEproof}

\begin{lemma}\it\label{l:123}
Let $\psi$ be an automorphism of $C$.
Then
\begin{enumerate}

\item[(i)] For every $j$ from $1$ to $\delta$,
 $\psi(w_j)$ is a linear combination
   of $y$, $w_1$, \ldots, $w_\delta$.

\item[(ii)] For every $i$ from $1$ to ${\doT\gamma}$,
  $\psi(u_i)$ is a linear combination
    of $y$, $w_1$, \ldots, $w_\delta$, $u_1$, \ldots, $u_{\doT\gamma}$.

\item[(iii)] For every $s\in \{0,1\}^\delta$,
  \begin{equation}\label{eq:alpha}
   \psi(v(s)) = v(\sigma(s)) + a y
                             + \sum_{j=1}^\delta c_j w_j
                             + \sum_{i=1}^{\doT\gamma} b_i u_i,
  \end{equation}
  where $\sigma$ is a permutation of the set $\{0,1\}^\delta$,
  $a=a(s) \in \{0,1\}$, $c_j=c_j(s) \in \{0,1\}$, $b_i=b_i(s) \in \{0,1\}$.

\item[(iv)] The permutation $\sigma$ is linear.

\item[(v)] If $\psi(u_i)= u_i$ and $\psi(w_j)= w_j$ for all $i \in \{1,\ldots, \doT\gamma\}$ and $j \in \{1,\ldots,\delta\}$, then the functions $a=a(s)$, $c_j=c_j(s)$, $b_i=b_i(s)$ are linear.
\end{enumerate}
\end{lemma}

\begin{IEEEproof}
Since $\psi$ is an automorphism of $\mathrm{Ker}(C)$ \cite{PheRif:2002},
the permuted matrix $\psi(K)$ remains a generator matrix for $\mathrm{Ker}(C)$.
So, (ii) is obvious. Taking into account Lemma~\ref{l:fix-block}, we see that (i) holds as well.

(iii) For every $s\in \{0,1\}^\delta$, since $\psi(v(s))$ belongs to $C=\bigcup_{s \in \{0,1\}^\delta} (\mathrm{Ker}(C)+v(s))$, it can be represented as $v(\sigma(s)) + ay  + \sum_{j=1}^\delta c_j w_j + \sum_{i=1}^{\doT\gamma} b_i u_i \in v(\sigma(s))+\mathrm{Ker}(C)$
 for some $\sigma(s) \in \{0,1\}^\delta$, $a,c_1,...,c_\delta,b_1,...,b_{\doT\gamma} \in\{0,1\}$.
 As $\psi$ is a bijective mapping and $\psi$ is an automorphism of $\mathrm{Ker} (C)$, the mapping $\sigma : \{0,1\}^\delta \rightarrow \{0,1\}^\delta$ is bijective too.

(iv) By Lemmas~\ref{l:v(s)} and~\ref{l:ker},
we have that $v(s)+v(s')+v(s'') \in \mathrm{Ker} (C)$ if and only if $s+s'+s''=0$.
From (\ref{eq:alpha}), we see that $\psi(v(s))+\psi(v(s'))+\psi(v(s'')) \in \mathrm{Ker} (C)$
if and only if $\sigma(s)+\sigma(s')+\sigma(s'')=0$.
The condition
$v(s)+v(s')+v(s'') \in \mathrm{Ker} (C)$ is readily the same as
$\psi(v(s))+\psi(v(s'))+\psi(v(s'')) \in \mathrm{Ker} (C)$ since $\psi \in \mathrm{Aut}(\mathrm{Ker}(C))$.
Thus, $s+s'+s''=0$ is equivalent to $\sigma(s)+\sigma(s')+\sigma(s'')=0$,
which means that $\sigma$ is linear.

(v) The hypothesis of statement (v) implies that $\psi$ does not permute the blocks.
By Lemma~\ref{l:v(s)}, for every block $I$ such that all but one $w_1$, \ldots, $w_\delta$
vanish on $I$, the mapping $s \mapsto v(s)|_{I}$ is linear.
Moreover, permuting the coordinates within the block does not change this property,
i.e., the mapping $s \mapsto \psi(v(s))|_{I}$ is linear too.
We first consider the first such block $I=\{1, \ldots,{2^\delta\}}$.
Since $w_1, \ldots,w_\delta$ and $u_1, \ldots,u_{\doT\gamma}$ vanish on $I$,
equation (\ref{eq:alpha}) turns to $\psi(v(s))|_{I} = v(\sigma(s))|_{I} + a(s) y|_{I}$.
We see that the mapping $s\mapsto a(s)$ must be linear.
Similarly, considering the block $I$ such that
all $w_1$, \ldots, $w_\delta$ and $u_1$, \ldots, $u_{\doT\gamma}$ except $u_i$ (or $w_j$) vanish on $I$,
we obtain the linearity of the mapping  $s \mapsto b_i(s)$ for $i=1, \ldots,{\doT\gamma}$ ($s \mapsto c_j(s)$
for $j=1, \ldots,\delta$, respectively).
\end{IEEEproof}

\begin{lemma}\it\label{l:3+}
Let $\psi$ be an automorphism of $C$ such that
$\psi(u_i)= u_i$ and $\psi(w_j)= w_j$ for all $i \in \{1,\ldots, \doT\gamma\}$ and $j \in \{1,\ldots,\delta\}$.
If $\delta\geq 3$, then $\sigma$ defined in Lemma~\ref{l:123} is the identity permutation.
\end{lemma}
\begin{IEEEproof}
The hypothesis of the lemma implies that $\psi$ stabilizes every block of coordinates.
The idea of the proof is to show that any nonidentical permutation $\sigma$ breaks the linear relations between the vectors $v(s)$ in a fixed block $I$ of coordinates.

We already know that
$\sigma:\{0,1\}^\delta \to \{0,1\}^\delta$
is a nonsingular linear transformation.
It remains to show that $\sigma$ fixes every basis vector $e_j$
that has $1$ in the $j$th position and $0$ otherwise.
We will prove this for $e_1$, as the other cases are similar.
First, let us prove for example that $\sigma(e_1)\ne t$
for every $t=(t_1, \ldots,t_\delta)$
such that $t_{2} =1$.
Consider a block $I$ such that all $w_1$, \ldots, $w_\delta$
vanish on $I$ except vectors $w_{2}$ and $w_{3}$,
for which $w_{2}|_I=w_{3}|_I = \overline{1}$
(where $\overline{1}$ is the all-one word of length~${2^\delta}$).
Note that $w_{3}$ exists because $\delta \geq 3$.
Then, by~Lemma~\ref{l:v(s)},
$$
v(s)|_I = \sum_{j=1}^{\delta}s_jv_j|_I+s_{2} s_{3} \overline{1},
 \quad s=(s_1, \ldots,s_\delta)\in\{0,1\}^\delta.
$$
It is straightforward to see that the mapping
$v^I(s)=v(s)|_I:\{0,1\}^\delta \to \{0,1\}^{2^\delta}$
is linear in the direction $e_1$,
since $v^I(s)+v^I(s+e_1)=v_1|_I$ is constant;
and it is nonlinear in the direction $t$,
as $v^I(s)+v^I(s+t)=x+(s_{2}t_{3}+s_{3})\overline 1$
is not constant, where $x$ is a constant vector.
However, we see from (\ref{eq:alpha}) that
$v^I(\sigma(s))$ is also linear in the direction $e_1$,
that is, $v^I(s)$ is linear in the direction $\sigma(e_1)$.
We see that $\sigma(e_1)\ne t$.
Similarly, this holds for every $t$ such that $t_{k}=1$
for some ${k}\ne 1$; so,
the only remaining possibility is
$\sigma(e_1)= e_1$.
Similarly, $\sigma(e_j)= e_j$ for every $j$;
hence, the linear map $\sigma$ is just the identity.
\end{IEEEproof}

\begin{example}
In the matrix $S_{0,3}$ of Example~\ref{ex:S}, we can see
that $v^I(s)+v^I(s+t)=v_1|_I$, so it does not depend on $s$
for the forth block $I=\{25,\ldots,32\}$ and for $t=(1,0,0)$.
This is not true for $t\ne (0,0,0), (1,0,0)$.
\end{example}

\begin{proposition}\it\label{c:ub}
If $\delta\geq 2$,
then the order of the automorphism group of $C$ satisfies
$$|\mathrm{Aut}(C)| \leq p\cdot
2^{\frac{1}{2}{\doT\gamma}({\doT\gamma}+1) +2{\doT\gamma}\delta + \frac{3}{2}\delta(\delta+1)}\prod_{i=1}^{{\doT\gamma}}(2^i-1)\prod_{j=1}^{\delta}(2^j-1),
$$
where $p=6$ if $\delta=2$ and $p=1$ if $\delta\geq 3$.
\end{proposition}
\begin{IEEEproof}
 Let $\psi\in \mathrm{Aut}(C)$.
 We evaluate the number of possibilities
 for the matrices $\psi(K)$ and $\psi(S)$.
 Clearly, $\psi(y)=y$.
 By Lemma~\ref{l:123}, for every $j$ from $1$ to $\delta$, $\psi(w_j)$ is a linear combination
 of $y$, $w_1$, \ldots, $w_\delta$.
 However, as $\psi(w_j)$ must be linearly independent on $y$, $\psi(w_1)$, \ldots, $\psi(w_{j-1})$,
 there are at most $2^{{\delta}+1}-2^j$ possibilities for $\psi(w_j)$, for each choice of $y$, $\psi(w_1)$, \ldots, $\psi(w_{j-1})$.
 Similarly, for every $i$ from $1$ to ${\doT\gamma}$, there are at most $2^{\delta+{\doT\gamma}+1}-2^{\delta+i}$ possibilities for $\psi(u_i)$.
In total, $\psi(K)$ is one of the
\begin{multline}\label{eq:Kways}
 \prod_{j=1}^{\delta}(2^{\delta+1}-2^j) \prod_{i=1}^{{\doT\gamma}} (2^{\delta+{\doT\gamma}+1}-2^{\delta+i})  \\[-1mm]
=2^{\frac{1}{2}{\doT\gamma}({\doT\gamma}+1) +{\doT\gamma}\delta + \frac{1}{2}\delta(\delta+1)}\prod_{i=1}^{{\doT\gamma}}(2^i-1)\prod_{j=1}^{\delta}(2^j-1)
\end{multline}
matrices obtained from $K$ by nonsingular transformations as above.

Now suppose that we know $\psi(K)$ and count the possibilities for $\psi(S)$.
W.l.o.g., $\psi(K)=K$. First, consider the case $\delta\geq 3$.
By Lemmas~\ref{l:123}(iii), \ref{l:123}(v) and \ref{l:3+},
all $\psi(v(s))$ are uniquely defined
by $1+\delta+{\doT\gamma}$ linear functions
$a$, $c_1$, \ldots, $c_\delta$, $b_1$, \ldots, $b_{\doT\gamma}$: $\{0,1\}^\delta \to \{0,1\}$.
Note that, in the current corollary, we do not state yet that every possibility is realizable by some $\psi$.
We totally have at most $(2^\delta)^{1+{\doT\gamma}+\delta}$ choices for $\psi(S)$, given $\psi(K)$.
Multiplying this last value by (\ref{eq:Kways}), we obtain the required bound.
Finally, if $\delta=2$, then the same arguments work with the only exception that there are additional possibilities to choose a nonsingular linear mapping $\sigma:\{0,1\}^2 \to \{0,1\}^2$, accordingly to Lemma~\ref{l:123}(iii) and \ref{l:123}(iv), which gives the factor $p=6=(2^2-1)(2^2-2)$.
\end{IEEEproof}

\section[The permutation automorphism group of a Z2Z4-linear Hadamard code]{The permutation automorphism group of a $\Z_2\Z_4$-linear Hadamard code}\label{s:paut}
\setstretch{1.00}
By Lemma \ref{l:13}, any nonsingular affine transformation of $\Z_2^{\doT\gamma} \times \Z_4^\delta$ belongs to $\mathrm{Aut}(C)$.
Therefore, since for $\delta\geq 3$, the number of nonsingular affine transformations coincides
with the upper bound given in Proposition~\ref{c:ub}, we obtain the following result:

\begin{theorem}\it\label{th:d3}
The permutation automorphism group of the $\Z_2\Z_4$-linear Hadamard code $C$ of type $(\alpha,\beta;\doT\gamma+1, \delta)$,
with $\delta\geq 3$, is the group of nonsingular affine transformations of $\Z_2^{\doT\gamma} \times \Z_4^\delta$. Therefore, its order is
 $$|\mathrm{Aut}(C)|=
2^{\frac{1}{2}{\doT\gamma}({\doT\gamma}+1) +2{\doT\gamma}\delta + \frac{3}{2}\delta(\delta+1)}\prod_{i=1}^{{\doT\gamma}}(2^i-1)\prod_{j=1}^{\delta}(2^j-1).
$$
\end{theorem}

It remains to consider the case $\delta = 2$.

\begin{lemma}\it\label{l:case2}
Let $C$ be the $\Z_2\Z_4$-linear Hadamard code of type $(\alpha,\beta;\doT\gamma+1, 2)$.
\begin{enumerate}
\item[(i)] Any nonsingular linear mapping $\sigma:\{0,1\}^2 \to \{0,1\}^2$, as defined in Lemma~\ref{l:123}(iii),
corresponds to an automorphism $\psi$ of $C$ that stabilizes every block.

\item[(ii)] Moreover, if this automorphism $\psi$ is not identical, then it is not an affine transformation of $\Z_2^{\doT\gamma} \times \Z_4^2$.
\end{enumerate}
\end{lemma}

\begin{IEEEproof}
(i)
There are exactly $p=6=(2^2-1)(2^2-2)$ such permutations.
For $\doT\gamma=0$, it is easy to check that the result is true.
The coordinate permutations $\psi$ corresponding to two permutations $\sigma$ are shown
on the following diagram; the other three nonidentity variants for $\sigma$ can be expressed from these two
using composition.
\begin{IEEEeqnarray*}{rCl}
S_{0,2}&=&
\left(\!\begin{array}{c@{}c@{}c@{}c@{\ }c@{}c@{}c@{}c@{\ }c@{}c@{}c@{}c@{\ }c@{}c@{}c@{}c@{\,}r}
0&0&0&0& 0&0&0&0& 0&0&0&0& 0&0&0&0\\
0&0&1&1& 0&0&1&1& 0&1&0&1& 0&1&0&1&\raisebox{-.7ex}[0.1ex][0.1ex]{$\Lsh$}\\
0&1&0&1& 0&1&0&1& 0&0&1&1& 0&1&1&0&\raisebox{ .7ex}[0.1ex][0.1ex]{\rotatebox[origin=c]{180}{$\Rsh$}}\\
0&1&1&0& 0&1&1&0& 0&1&1&0& 1&1&0&0\\
& \makebox[0.8ex][l]{\rotatebox[origin=c]{90}{$\Rsh$}}&\makebox[0.8ex][r]{\rotatebox[origin=c]{270}{$\Lsh$}} &&
& \makebox[0.8ex][l]{\rotatebox[origin=c]{90}{$\Rsh$}}&\makebox[0.8ex][r]{\rotatebox[origin=c]{270}{$\Lsh$}} &&
& \makebox[0.8ex][l]{\rotatebox[origin=c]{90}{$\Rsh$}}&\makebox[0.8ex][r]{\rotatebox[origin=c]{270}{$\Lsh$}} &&
&& \makebox[0.8ex][l]{\rotatebox[origin=c]{90}{$\Rsh$}}&\makebox[0.8ex][r]{\rotatebox[origin=c]{270}{$\Lsh$}} &
\end{array}\!\!\right) 
\\&=&
\!\left(\begin{array}{c@{}c@{}c@{}c@{\ }c@{}c@{}c@{}c@{\ }c@{}c@{}c@{}c@{\ }c@{}c@{}c@{}c@{\,}r}
0&0&0&0& 0&0&0&0& 0&0&0&0& 0&0&0&0\\
0&0&1&1& 0&0&1&1& 0&1&0&1& 0&1&0&1\\
0&1&0&1& 0&1&0&1& 0&0&1&1& 0&1&1&0&\raisebox{-.7ex}[0.1ex][0.1ex]{$\Lsh$}\\
0&1&1&0& 0&1&1&0& 0&1&1&0& 1&1&0&0&\raisebox{ .7ex}[0.1ex][0.1ex]{\rotatebox[origin=c]{180}{$\Rsh$}}\\
&& \makebox[0.8ex][l]{\rotatebox[origin=c]{90}{$\Rsh$}}&\makebox[0.8ex][r]{\rotatebox[origin=c]{270}{$\Lsh$}} &
&& \makebox[0.8ex][l]{\rotatebox[origin=c]{90}{$\Rsh$}}&\makebox[0.8ex][r]{\rotatebox[origin=c]{270}{$\Lsh$}} &
& \makebox[0.8ex][l]{\rotatebox[origin=c]{90}{$\Rsh$}}&&\makebox[0.8ex][r]{\rotatebox[origin=c]{270}{$\Lsh$}} &
 \makebox[0.8ex][l]{\rotatebox[origin=c]{90}{$\Rsh$}}&&\makebox[0.8ex][r]{\rotatebox[origin=c]{270}{$\Lsh$}} &&
\end{array}\!\!\right)
\end{IEEEeqnarray*}
For $\doT\gamma\ge 1$,
the $4$ rows $v(s)$  of $S_{\doT\gamma,2}$, $s\in \{0,1\}^2$, contain in each one of their $4$ macroblocks of size $2^{\doT\gamma +2}$,
the same coordinates corresponding to the same macroblock in the row $v(s)$ of $S_{0,2}$ repeated $2^{\doT\gamma}$ times. Therefore,
the result follows.

(ii) We first note that, as follows from the first part of the proof,
the considered automorphisms fix a coordinate from every block.
Now, let us consider an arbitrary affine permutation $\psi$ that fixes a coordinate from every block
and show that it is the identity permutation $\mathrm{id}$.
Assume, seeking a contradiction, that $\psi\ne\mathrm{id}$.
Consider the map $\xi = \psi - \mathrm{id}$.
It is an affine map from $\Z_2^{\doT\gamma} \times \Z_4^2$ to
$\Z_2^{\doT\gamma} \times \Z_4^2$.
It possesses the zero value in a point from every block and a nonzero value in at least one point.
Consider a nonconstant linear mapping $\lambda$
from the image of $\xi$ onto $\{0,2\}$.
Then, the map $\lambda\xi$ from $\Z_2^{\doT\gamma} \times \Z_4^2$ to $\{0,2\}$
is affine and is not constantly zero.
By Lemma~\ref{l:ker}, $\lambda\xi$ belongs to the kernel of $C$.
However, as $\lambda\xi$ is zero in a point from every block, by the definition of the blocks, it is constantly zero.
We get a contradiction.
\end{IEEEproof}

\begin{theorem}\it\label{th:d3-2}
The automorphism group of the $\Z_2\Z_4$-linear Hadamard code $C$ of type $(\alpha,\beta;\doT\gamma+1, 2)$
consists of all permutations expressed as $\psi\alpha$, where $\alpha$
is a nonsingular affine transformation of $\Z_2^{\doT\gamma} \times \Z_4^2$
and $\psi$ is one of the six permutations defined in Lemma~\ref{l:case2}.
The order of this automorphism group is
 $$|\mathrm{Aut}(C)|=
6\cdot 2^{\frac{1}{2}{\doT\gamma}({\doT\gamma}+1) +4{\doT\gamma} +
9}\cdot 3 \prod_{i=1}^{{\doT\gamma}}(2^i-1).
$$
\end{theorem}
\begin{IEEEproof}
By Lemma~\ref{l:case2}(i), all the permutations are automorphisms of $C$.
By Lemma~\ref{l:case2}(ii), they are mutually different.
We can see that the lower bound given by counting all such permutations coincides with the upper bound
from Proposition~\ref{c:ub}.
\end{IEEEproof}

\setstretch{1.13}
\enlargethispage{-56mm}

\providecommand\href[2]{#2} \providecommand\url[1]{\href{#1}{#1}}
  \def\DOI#1{{\small {DOI}:
  \href{http://dx.doi.org/#1}{#1}}}\def\DOIURL#1#2{{\small{DOI}:
  \href{http://dx.doi.org/#2}{#1}}}

\end{document}